\documentclass[reqno]{amsart}

\usepackage{a4wide, mathrsfs}
\usepackage{hyperref}

\title{Areas and volumes for null cones}

\numberwithin{equation}{section}

\theoremstyle{plain}
\newtheorem{theorem}{Theorem}[section]
\newtheorem{corollary}[theorem]{Corollary}
\newtheorem{proposition}[theorem]{Proposition}
\newtheorem{lemma}[theorem]{Lemma}

\theoremstyle{definition}
\newtheorem{definition}[theorem]{Definition}
\newtheorem{example}[theorem]{Example}

\theoremstyle{remark}
\newtheorem{remark}[theorem]{Remark}

\newcommand\g{\ensuremath{\mathbf{g}}}
\newcommand\T{\ensuremath{\mathbf{T}}}

\newcommand\bs{\ensuremath{\boldsymbol{\sigma}}}
\newcommand\bS{\ensuremath{\mathbf{S}}}

\renewcommand\l{\ensuremath{\boldsymbol{\ell}}}
\newcommand\n{\ensuremath{\mathbf{n}}}
\newcommand\m{\ensuremath{\mathbf{m}}}
\newcommand\mbar{\ensuremath{\overline{\m}}}

\newcommand\Ric{\ensuremath{\mathbf{Ric}}}
\newcommand\R{\ensuremath{\mathbb{R}}}
\newcommand\sn{\ensuremath{\mathrm{sn}}}
\newcommand\Id{\ensuremath{\mathrm{Id}}}

\newcommand\Ncal{\ensuremath{\mathcal{N}}}
\newcommand\Scal{\ensuremath{\mathcal{S}}}

\newcommand\eps{\ensuremath{\varepsilon}}

\newcommand\bel[1]{\begin{equation}\label{#1}}
\newcommand\ee{\end{equation}}

\newcommand\bracket[2]{\ensuremath{\left\{ #1 \, \left| \vphantom{|^|} \, #2 \right.\right\}}}

\DeclareMathOperator{\tr}{tr}

\newcommand\vs{\vskip .2cm}

\newcommand\gammad{\ensuremath{{\gamma}^{\prime}}}

\begin{document}
\author{James~D.E.\ Grant}
\address{Institut f\"{u}r Grundlagen der Bauingenieurwissenschaften\\
Leopold-Franzens-Universit\"{a}t Innsbruck\\
Technikerstrasse 13\\ 6020 Innsbruck\\ Austria}
\address{Current address: \href{http://www.mat.univie.ac.at/home.php}{Fakult{\"a}t f{\"u}r Mathematik} \\
\href{http://www.univie.ac.at/de/}{Universit{\"a}t Wien} \\ Nordbergstrasse 15 \\ 1090 Wien \\ Austria}
\email{\href{mailto:james.grant@univie.ac.at}{james.grant@univie.ac.at}}
\urladdr{\href{http://jdegrant.wordpress.com}{http://jdegrant.wordpress.com}}
\thanks{This work was supported by START-project Y237--N13 of the \href{http://www.fwf.ac.at/}{Austrian Science Fund} and by the Agence Nationale de la Recherche (ANR) through the grant 06-2-134423, ``Mathematical Methods in General Relativity'' (MATH-GR). The author is grateful to Universit\'{e} Pierre et Marie Curie (Paris 6) for their hospitality during the completion of this work, and to Prof.~P.\ Chru{\'s}ciel for comments on a preliminary version of this paper.}

\subjclass[2010]{53C23, 53C80}
\keywords{Volume comparison, area comparison, monotonicity properties}

\begin{abstract}
Motivated by recent work of Choquet-Bruhat, Chru{\'s}ciel, and Mart{\'{\i}}n-Garc{\'{\i}}a~\cite{LightCone}, we prove monotonicity properties and comparison results for the area of slices of the null cone of a point in a Lorentzian manifold. We also prove volume comparison results for subsets of the null cone analogous to the Bishop--Gromov relative volume monotonicity theorem and G\"{u}nther's volume comparison theorem. We briefly discuss how these estimates may be used to control the null second fundamental form of slices of the null cone in Ricci-flat Lorentzian four-manifolds with null curvature bounded above.
\end{abstract}
\maketitle
\thispagestyle{empty}

\section{Introduction}
The application of comparison techniques to problems in Riemannian geometry is now well-established. More recently, there has been a significant application of comparison-theoretic machinery to specific problems in Lorentzian geometry, such as volume comparison theorems, and related rigidity results.\footnote{See, e.g.,~\cite{EhrlichKim} for a recent review. In addition, our approach was significantly influenced by~\cite{ES}.} A new type of comparison theorem in Lorentzian geometry was given in a recent paper~\cite{LightCone}, where the authors showed that the area of the cross-sections of a light-cone in a Lorentzian manifold satisfying the Dominant Energy condition are bounded above by areas of corresponding sections in Minkowski space. This result is reminiscent of the area and volume comparison theorems in Riemannian geometry, such as the Bishop comparison theorem, where one compares the volume of a metric ball in a Riemannian manifold with Ricci curvature bounded below with the volume of a ball of the same radius in the corresponding constant curvature space. The current paper arose from the wish to generalise the considerations of~\cite{LightCone} by developing null analogues of other Riemannian comparison results. We first show that the results of~\cite{LightCone} may, in one sense, be strengthened, to show that the ratio of the area of cross-sections of the null cone in a manifold with curvature bounded below to a specific quantity determined in terms of the curvature bound satisfies a monotonicity property. The result of~\cite{LightCone} arises as a special case of this monotonicity result. Using a simple result from~\cite{CGT}, we then make the simple deduction that this area monotonicity result leads to a relative null volume monotonicity result analogous to the Bishop--Gromov volume comparison theorem.

In an alternative direction, we show that, assuming an \emph{upper bound\/} on the null curvature along the null cone, one may deduce an alternative area-monotonicity result, which gives a \emph{lower bound\/} on the cross-sectional area of the light-cone. Integrating this theorem gives a lower bound on the null volume of a subset of the null cone. This result may, in essence, be viewed as an analogue of G\"{u}nther's volume comparison theorem in Riemannian geometry. Unlike the case with curvature bounded below, this result requires the analysis of a matrix Riccati equation, rather than a scalar Riccati equation.

Finally, we briefly investigate some model Lorentzian geometries for which our comparison results are sharp. Unlike many standard comparison constructions, our model geometries are not unique, and we do not have rigidity results in the cases where our inequalities are saturated.\footnote{At least, not without the imposition of additional conditions on the model geometries.} We also briefly discuss how our results may be used to control the mean curvature of the slices of a null cone in a four-dimensional Ricci-flat four-manifold in terms of the \lq\lq area radius\rq\rq\ and \lq\lq volume radius\rq\rq.

\vs
This paper is organised as follows. In the following section, we recall basic material concerning the geometry of null cones. In Section~\ref{sec:Riccati}, we develop Riccati equation techniques that allow us to estimate the null second fundamental form of a slice of the null cone under various types of curvature bound. In Section~\ref{sec:comparison}, the results of this Riccati equation analysis are applied to derive a monotonicity result for the area of a slice of the null cone. {}From this result, we directly derive a volume monotonicity result, somewhat analogous to the Bishop--Gromov volume comparison result. Both of these results require a lower bound on the Ricci tensor along the null cone. Assuming an upper bound on the curvature along the null cone, we derive a corresponding area monotonicity result, and an analogue of the G\"{u}nther volume comparison result. In Section~\ref{sec:4d}, we discuss an application of our results to the estimation of the null mean curvature of spheres in terms of the \lq\lq area radius\rq\rq\ and \lq\lq volume radius\rq\rq\ for four-dimensional, Ricci-flat metrics. In Section~\ref{sec:model}, we recast our results in terms of model geometries, both Riemannian and Lorentzian. Finally, for the convenience of readers familiar with this notation, we outline in an appendix how our results appear in four dimensions, when carried out in Newman--Penrose formalism. With the exception of this appendix, this paper is essentially self-contained.

\section{Background material and notation}
\label{sec:basic}

Let $(M, \g)$ be a smooth, time-oriented Lorentzian manifold of dimension $n+1$, with the metric $\g$ having signature $(-, +, \dots, +)$. We assume that $(M, \g)$ is geodesically complete. Let $p \in M$ and let $\Ncal^+(p)$ denote the future null cone of the point $p$. Given a unit-length, future-directed, time-like vector $\T \in T_p M$, we define $S^+_1(0) \subseteq T_p M$ as the set of future-directed, null vectors $\l \in T_p M$ that satisfy the normalisation condition
\bel{nullnormalisation}
\g(\T, \l) = -1.
\ee
Given $\l \in S^+_1(0)$, we denote by $\gamma_{\l} \colon [0, \infty) \to M$ the future-directed, affinely-parametrised geodesic such that $\gamma_{\l}(0) = p$, ${\gamma}^{\prime}_{\l}(0) = \l$. We define
\[
\Scal_s := \bracket{\gamma_{\l}(s)}{\l \in S^+_1(0)},
\]
and the set
\[
\Ncal^+_s(p) := \bigcup_{0 \le t \le s} \Scal_s.
\]
Except briefly in \S7, we will assume that $s > 0$ is less than the null injectivity radius at $p$, in which case the set $\Scal_s$ is a smoothly embedded $(n-1)$-dimensional sphere in $M$ and $\Ncal^+_s(p) \subset \Ncal^+(p)$. The sphere $\Scal_s$ inherits an induced Riemannian metric, which we denote by $\bs_s$. We denote the area of the set $\Scal_s$ with respect to the metric $\bs_s$ by
\[
|\Scal_s|_{\g} = \int_{\Scal_s} \, dV_{\bs_s}.
\]

In a slight abuse of notation, we will also use $\l$ to denote the tangent vector field ${\gamma}^{\prime}_{\l}$ defined on the set $\Ncal^+_s(p)$. Given a tensor or scalar field defined along the null cone, $\rho$, we will denote its covariant derivative along the null geodesics that generate $\Ncal^+_s(p)$ by $\rho' \equiv \nabla_{\l} \rho \equiv \nabla_{{\gamma}^{\prime}_{\l}} \rho$. In terms of the vector field $\l$, we define the null shape operator of $\Scal_s$, $\bS \colon \mathfrak{X}(\Scal_s) \to \mathfrak{X}(\Scal_s)$ by\footnote{Throughout, we will use the notation $\langle \mathbf{u}, \mathbf{v} \rangle \equiv \g(\mathbf{u}, \mathbf{v})$ to refer to the inner product with respect to the Lorentzian metric $\g$.}
\[
\langle \mathbf{u}, \bS(\mathbf{v}) \rangle := \langle \mathbf{u}, \nabla_{\mathbf{v}} \l \rangle,
\]
and the corresponding null mean curvature
\[
H := \frac{1}{n-1} \tr\bS.
\]
A standard result is that the derivative with respect to $s$ of this area is given by
\bel{ddsSs}
\frac{d}{ds} |\Scal_s|_{\g} = \int_{\Scal_s} \tr\bS \, dV_{\bs_s} = (n-1) \int_{\Scal_s} H \, dV_{\bs_s}.
\ee

\begin{example}
In flat Minkowski space $\R^{n, 1}$, letting $p$ lie at the origin, the sphere $\Scal_s$ is the set
\[
\Scal_s = \left\{ t = r = s \right\},
\]
with area
\[
|\Scal_s|_{\g} = \omega_{n-1} s^{n-1},
\]
where $\omega_{n-1}$ denotes the area of the unit sphere in $\R^n$. A straightforward calculation yields that
\[
\bS(\Scal_s) = \frac{1}{s} \, \Id_{T \Scal_s}, \qquad H(\Scal_s) = \frac{1}{s}.
\]
These expressions will also give the limiting form of $\bS(\Scal_s)$ and $H(\Scal_s)$ as $s \to 0$ in an arbitrary Lorentzian manifold.
\end{example}

\section{Riccati techniques}
\label{sec:Riccati}

We will now develop some techniques that we will require to prove our comparison results.

\begin{definition}
Let $q \in M$. A \emph{null basis\/} at $q$ is a basis $(\l, \n, \mathbf{e}_1, \dots, \mathbf{e}_{n-1})$ for $T_q M$ with the property that
\bel{nullbasis}
\langle \l, \n \rangle = - 2, \qquad \langle \mathbf{e}_i, \mathbf{e}_j \rangle = \delta_{ij},
\ee
with other products vanishing. By a null basis on a connected set, we will mean a smoothly varying null basis at each point of the set.
\end{definition}

\begin{lemma}
\label{lemma:basis}
Given any point $q \in \Ncal^+_s(p) \setminus \{ p \}$, we may choose a null basis on a neighbourhood (in $\Ncal^+_s(p) \setminus \{ p \}$) of $q$ with the properties that
\bel{BasisDeriv}
\nabla_{\l} \l = 0,
\qquad
\nabla_{\l} \mathbf{e}_i = \alpha_i \l,
\qquad
\nabla_{\l} \n = 2 \alpha_i \mathbf{e}_i.
\ee
\end{lemma}

\begin{proof}
Given a normalised null vector $\l \in S^+_1(0)$, the affinely-parametrised geodesic $\gamma_{\l}$ is uniquely determined. By assumption, the geodesics $\gamma_{\l}$ are affinely-parametrised, and therefore satisfy
\[
\nabla_{\gammad_{\l}} \gammad_{\l} = 0.
\]
As $\l$ varies in $S^+_1(0)$, the tangent vectors $\gammad_{\l}$ determine a unique vector field on the set $\Ncal^+_s(p) \setminus \{ p \}$. As before, we will denote this vector field by $\l$. We then have the rank-$n$ vector bundle $\l^{\perp} \subset \left. TM \right|\Ncal^+_s(p) \setminus \{ p \}$. Given $q \equiv \gamma_{\l}(s_q) \in \Ncal^+_s(p) \setminus \{ p \}$, the fibre of this bundle is spanned by the vector $\gammad_{\l}(s_q)$ along with the tangent space, $T_{\gamma_{\l}(s_q)} \Scal_{s_q}$, to the sphere $\Scal_{s_q}$. We fix an orthonormal basis $\{ \mathbf{e}_1, \dots, \mathbf{e}_{n-1} \}$ of $T_{\gamma_{\l}(s_q)} \Scal_{s_q}$. The null orthogonality conditions~\eqref{nullbasis} now uniquely determine the null vector $\n(s_q) \in T_{\gamma_{\l}(s_q)}(s) M$ conjugate to $\gammad_{\l}(s_q)$.

We repeat this construction at each point of an open neighbourhood of $q$, giving a smooth basis $\{ \l, \n, \mathbf{e}_1, \dots, \mathbf{e}_{n-1} \}$. By construction, the distribution spanned by $\{ \l, \mathbf{e}_i \}$ is integrable. In addition, the distribution spanned by the $\{ \mathbf{e}_i \}$ is integrable, thereby ensuring that the operator $\bS$ is symmetric. Finally, the orthogonality relationships imply that
\[
\nabla_{\l} \mathbf{e}_i = \alpha_i \l + \beta_{ij} \mathbf{e}_j,
\]
for some functions $\alpha_i$ and $\beta_{ij}$, where $\beta_{ij} = - \beta_{ji}$. If we perform an orthogonal transformation to another basis $\widetilde{\mathbf{e}_i} = \Lambda_{ij} \mathbf{e}_j$, where $\Lambda \in \mathrm{SO}_{n-1}$, then we find that
\[
\widetilde{\beta} = \nabla_{\l} \Lambda + \Lambda \beta.
\]
Taking $\Lambda \colon [0, s] \to \mathrm{SO}_{n-1}$ to satisfy the ordinary differential equation
\[
\nabla_{\l} \Lambda(s) + \Lambda(s) \beta(s) = 0, \qquad \Lambda(s) \to \Id_{n-1} \mbox{ as $s \to 0$},
\]
we may ensure that $\widetilde{\beta} = 0$ along the geodesics $\gamma_{\l}$. Dropping the tilde's, the vector fields $\{ \l, \mathbf{e}_1, \dots, \mathbf{e}_{n-1} \}$ satisfy the required stated in~\eqref{BasisDeriv}. The form of the derivative of the complementary null vector $\n$ now follows from the preservation of the null-orthogonality conditions along the geodesics $\gamma_{\l}$.
\end{proof}

Given a point $\gamma_{\l}(s) \in \Scal_s$, we denote by $P \colon T_{\gamma_{\l}(s)} M \to T_{\gamma_{\l}(s)} \Scal_s$ the orthogonal projection onto the tangent space to the sphere, $\Scal_s$, at $\gamma_{\l}(s)$. In terms of the local basis introduced above, this map is written in the form $\mathbf{v} \mapsto \langle \mathbf{v}, \mathbf{e}_i \rangle \mathbf{e}_i$ for $\mathbf{v} \in T_{\gamma_{\l}(s)} M$.

\begin{definition}
Let $\l \in S^+_1(0)$, and $\gamma_{\l}$ the corresponding null geodesic. For $s > 0$, we define the map
\[
\mathcal{R}_{\l}(\gamma_{\l}(s)) \colon T_{\gamma_{\l}(s)} \Scal_s \to T_{\gamma_{\l}(s)} \Scal_s; \quad \mathbf{v} \mapsto P \left( \mathbf{R}(\mathbf{v}, \l) \l \right),
\]
and denote the corresponding operator along the geodesic $\gamma_{\l}$ by $\mathcal{R}_{\l}$.
\end{definition}

\begin{proposition}
The covariant derivative of the null shape operator, $\bS$, along the geodesic $\gamma_{\l}$ satisfies the identity
\bel{SRiccati}
\nabla_{\l} \bS = - \mathcal{R}_{\l} - \bS^2.
\ee
\end{proposition}
\begin{proof}
The result is local, so we may calculate using the basis $(\l, \n, \mathbf{e}_i)$ introduced in Lemma~\ref{lemma:basis}. We have
\begin{align*}
\langle \mathbf{e}_i, \left( \nabla_{\l}\bS \right) (\mathbf{e}_j) \rangle &=
\nabla_{\l} \langle \mathbf{e}_i, \bS(\mathbf{e}_j) \rangle
= \langle \mathbf{e}_i, \nabla_{\l} \nabla_{\mathbf{e}_j} \l \rangle
\\
&=
\langle \mathbf{e}_i, \mathbf{R}(\l, \mathbf{e}_j) \l + \nabla_{\left[ \l, \mathbf{e}_j \right]} \l \rangle
\\
&=
- \langle \mathbf{e}_i, \mathcal{R}_{\l}(\mathbf{e}_j) \rangle
+ \langle \mathbf{e}_i, \nabla_{\left[ \l, \mathbf{e}_j \right]} \l \rangle.
\end{align*}
In addition,
\begin{align*}
\left[ \l, \mathbf{e}_j \right]
&= \nabla_{\l} \mathbf{e}_j - \nabla_{\mathbf{e}_j} \l
= \alpha_j \l - \nabla_{\mathbf{e}_j} \l
= \left( \alpha_j + \frac{1}{2} \langle \nabla_{\mathbf{e}_j} \l, \n \rangle \right) \l
- \langle \nabla_{\mathbf{e}_j} \l, \mathbf{e}_k \rangle \mathbf{e}_k
\\
&=
\left( \alpha_j + \frac{1}{2} \langle \nabla_{\mathbf{e}_j} \l, \n \rangle \right) \l
- \langle \mathbf{e}_k, \bS(\mathbf{e}_j) \rangle \mathbf{e}_k.
\end{align*}
Therefore,
\[
\langle \mathbf{e}_i, \nabla_{\left[ \l, \mathbf{e}_j \right]} \l \rangle
= - \langle \mathbf{e}_i, \bS(\mathbf{e}_k) \rangle \langle \mathbf{e}_k, \bS(\mathbf{e}_j) \rangle
= - \langle \mathbf{e}_i, \bS^2(\mathbf{e}_j) \rangle,
\]
as required.
\end{proof}

Equation~\eqref{SRiccati}, along with the boundary condition that $s \cdot \bS(s) \to \Id$ as $s \to 0$, is a starting point for deriving area comparison and volume monotonicity results. For convenience, we define the following comparison functions (cf., e.g.,~\cite{Petersen}). Given $K \in \R$, we define
\bel{sn}
\sn_K(s) :=
\begin{cases}
\frac{1}{\sqrt{K}} \sin (\sqrt{K} s), &K> 0,
\\
s, &K = 0,
\\
\frac{1}{\sqrt{|K|}} \sinh (\sqrt{|K|} s), &K < 0.
\end{cases}
\ee
We then have the following comparison results:

\begin{proposition}
\label{prop:RiccComp}
Let $c$ be a real constant such that $\Ric(\gammad_{\l}, \gammad_{\l}) \ge c (n-1)$ along the geodesic $\gamma_{\l}$. Then
\bel{RicciComp}
\tr\bS(\gamma_{\l}(s)) \le (n-1) \frac{\sn^{\prime}_c(s)}{\sn_c(s)}, \qquad s > 0.
\ee
Alternatively, let $K$ be a real constant such that $\mathcal{R}_{\l}(\gamma_{\l}(s)) \le K \, \Id_{T_{\gamma_{\l}(s)} \Scal_s}$ along $\gamma_{\l}$.\footnote{By this, we mean that the eigenvalues of $\mathcal{R}_{\l}$ are bounded above by $K$, so $\langle \mathbf{v}, \mathcal{R}_{\l}(\mathbf{v}) \rangle \le K \langle \mathbf{v}, \mathbf{v} \rangle$ for all $\mathbf{v} \in T_{\gamma_{\l}(s)} \Scal_s$ along the geodesic $\gamma_{\l}$.} Then
\bel{SectComp}
\bS(\gamma_{\l}(s)) \ge \frac{\sn^{\prime}_K(s)}{\sn_K(s)} \, \Id_{T_{\gamma_{\l}(s)} \Scal_s}.
\ee
In particular,
\bel{TrSectComp}
\tr\bS(\gamma_{\l}(s)) \ge (n-1) \frac{\sn^{\prime}_K(s)}{\sn_K(s)}.
\ee
\end{proposition}
\begin{proof}
For completeness, we give proofs of both results even though they are adaptions of quite standard techniques.

\vs
For simplicity, we denote quantities such as $\bS(\gamma_{\l}(s))$ by $\bS(s)$ for the duration of the proof. Let $H(s) := \frac{1}{n-1} \tr\bS(s)$. It follows from the asymptotics of $\bS(s)$ that $s \cdot H(s) \to 1$ as $s \to 0$. We now note that, applying the Cauchy-Schwarz inequality for $(n-1) \times (n-1)$ symmetric matrices, we have that
\bel{CS}
H^2 = \frac{1}{(n-1)^2} \left( \tr \bS \right)^2 \le \frac{1}{n-1} \tr \bS^2.
\ee
Taking the trace of~\eqref{SRiccati}, and substituting the inequality~\eqref{CS}, we deduce that $H$ satisfies the differential inequality
\[
H^{\prime}(s) + H(s)^2
\le - \frac{1}{n-1} \tr \mathcal{R}_{\l}(s).
\]
Letting $\{ \mathbf{e}_i \}_{i=1}^{n-1}$ denote any orthonormal basis for $T_q M$ (at any point $q$ of interest to us), then
\begin{align*}
\tr \mathcal{R}_{\l} &=
\sum_{i=1}^{n-1} \langle \mathbf{e}_i, \mathcal{R}_{\l}(\mathbf{e}_i) \rangle =
\sum_{i=1}^{n-1} \langle \mathbf{e}_i,
\mathbf{R}(\mathbf{e}_i, \gammad_{\l}(s)) \gammad_{\l}(s) \rangle
\\
&=
\Ric(\gammad_{\l}(s) , \gammad_{\l}(s)) +
\frac{1}{2} \mathbf{R}(\gammad_{\l}(s), \gammad_{\l}(s), \gammad_{\l}(s), \n(s))
\\
&\hskip 4.5cm+
\frac{1}{2} \mathbf{R}(\gammad_{\l}(s), \n(s), \gammad_{\l}(s), \gammad_{\l}(s))
\\
&=
\Ric(\gammad_{\l}(s) , \gammad_{\l}(s))
\ge c (n-1).
\end{align*}
Therefore, $H$ satisfies the inequality
\bel{Hinequality}
H^{\prime}(s) + H(s)^2 \le - c.
\ee
Let $H(s) = \frac{a'(s)}{a(s)}$, with $a(0) = 0, a'(0) = 1$. We then have
\[
a^{\prime\prime}(s) + c \, a(s) \le 0.
\]
We now note that the comparison function $\sn_c(s)$ satisfies the differential equation
\[
{\sn}^{\prime\prime}_c(s) + c \, \sn_c(s) = 0,
\]
with the same boundary conditions as $a$ at $s = 0$. It follows that
\[
\frac{d}{ds} \left( a'(s) \sn_c(s) - \sn^{\prime}(s) a(s) \right) \le 0,
\]
so the quantity $a'(s) \sn_c(s) - \sn^{\prime}(s) a(s)$ is non-increasing as a function of $s$. Since this quantity is zero at $s = 0$, we deduce that $a'(s) \sn_c(s) \le \sn^{\prime}(s) a(s)$ for $s > 0$. Therefore
\[
\tr\bS(s) = (n-1) H(s) = (n-1) \frac{a'(s)}{a(s)} \le (n-1) \frac{\sn^{\prime}_c(s)}{\sn_c(s)},
\]
as required.

\

For the second result, we must use the full matrix Riccati equation~\eqref{SRiccati}. We follow the technique of~\cite[Chapter~6]{Petersen}.

The operator $\bS(s)$ is symmetric on $T_{\gamma_{\l}(s)} \Scal_s$ with respect to the inner product $\left. \bs_s \right|_{\gamma_{\l}(s)}$. It therefore has real eigenvalues, which we label as $\lambda_1(s) \le \dots \le \lambda_{n-1}(s)$. Since $\bS(s)$ is smooth in $s$, a $\min$-$\max$ argument implies that these eigenvalues are Lipschitz functions of $s$, and are smooth when the eigenvalues are distinct. We assume, for simplicity, that the eigenvalues are smooth.\footnote{The case where the eigenvalues are Lipschitz may be treated by barrier methods.} Finally, note that, since $s \cdot \bS(s) \to \mathrm{Id}$ as $s \to 0$, the eigenvalues satisfy the asymptotic condition that $s \cdot \lambda_i(s) \to 1$ as $s \to 0$ for $i = 1, \dots, n-1$.

Let $t > 0$ be fixed, with $\lambda_1(t)$ the lowest eigenvalue of $\bS(t)$ with corresponding unit-length eigenvector $\mathbf{v}_1(t) \in T_{\gamma_{\l}(t)} \Scal_t$. Then there exist coefficients $a_1, \dots, a_{n-1}$ such that
\[
\mathbf{v}_1(t) = \sum_{i=1}^{n-1} a_i \, \mathbf{e}_i(t).
\]
For $s > 0$, let $\mathbf{V}$ be the vector field along $\gamma_{\l}$ defined by
\[
\mathbf{V}(s)  = \sum_{i=1}^{n-1} a_i \, \mathbf{e}_i(s).
\]
We define the function
\[
\Lambda_1(s) := \langle \mathbf{V}(s), \bS(s) (\mathbf{V}(s)) \rangle, \qquad s > 0.
\]
A $\min$-$\max$ argument then implies that
\bel{minmax}
\Lambda_1(s) \ge \lambda_1(s)
\ee
for all $s > 0$, with equality when $s = t$. Since $\lambda_1(s)$ and $\Lambda_1(s)$ are smooth at $s = t$, and~\eqref{minmax} holds for all $s$ on a neighbourhood of $t$, it follows that ${\Lambda}^{\prime}_1(t) = {\lambda}^{\prime}_1(t)$. We therefore have
\begin{align*}
{\lambda}^{\prime}_1(t) &= \left. \frac{d}{ds} \Lambda_1(s) \right|_{s=t}
\\
&=
\left. \nabla_{\gamma^{\prime}_{\l}(s)} \left< \mathbf{V}(s), \bS(s) (\mathbf{V}(s)) \right> \right|_{s=t}
\\
&=
\left. \left< \mathbf{V}(s), \bS^{\prime}(s) (\mathbf{V}(s)) \right> \right|_{s=t}
+ \left< \left. \mathbf{V}^{\prime}(s) \right|_{s=t}, \bS(t) (\mathbf{v}_1(t)) \right>
\\
&\hskip 5.5cm + \left< \mathbf{v}_1(t), \bS(t) (\left. \mathbf{V}^{\prime}(s) \right|_{s=t}) \right>
\\
&=
\left< \mathbf{v}_1(t), \left[ - \bS(t)^2 - \mathcal{R}_{\l} \right] \mathbf{v}_1(t) \right>
+ 2 \lambda_1(t) \left< \mathbf{V}^{\prime}(t), \mathbf{v}_1(t) \right>
\\
&=
- \lambda_1(t)^2 - \left< \mathbf{v}_1(t), \mathcal{R}_{\l} (\mathbf{v}_1(t)) \right>
+ 2 \lambda_1(t) \left< \mathbf{V}^{\prime}(t), \mathbf{v}_1(t) \right>,
\end{align*}
where the fourth equality follows from~\eqref{SRiccati} and the symmetry of the operator $\bS(t)$ with respect to the inner product. We now note that
\[
\left< \mathbf{V}^{\prime}(t), \mathbf{v}_1(t) \right> =
\sum_{i, j = 1}^{n-1} a_i a_j \alpha_i \left< \l, \mathbf{e}_j \right> = 0.
\]
Therefore, imposing the curvature bound $\mathcal{R}_{\l} \le K \, \Id$, we have
\[
{\lambda}^{\prime}_1(t)
= - \lambda_1(t)^2 - \left< \mathbf{v}_1(t), \mathcal{R}_{\l} (\mathbf{v}_1(t)) \right>
\ge - \lambda_1(t)^2 - K.
\]
Changing the variable back from $t$ to $s$, we therefore have that, for all $s > 0$, the inequality
\[
{\lambda}^{\prime}_1(s) \ge - \lambda_1(s)^2 - K
\]
holds. Letting $\lambda_1(s) = \frac{a'(s)}{a(s)}$ with $a(0) = 0$, $a'(0) = 1$, we deduce that
\[
a^{\prime\prime}(s) + K a(s) \ge 0.
\]
Proceeding as in the proof of the first result, we conclude that
\[
\lambda_1(s) \ge \frac{\sn^{\prime}_K(s)}{\sn_K(s)}.
\]
Since $\lambda_1(s)$ is the lowest eigenvalue of $\bS(s)$, this inequality implies the required result~\eqref{SectComp}. Taking the trace of~\eqref{SectComp} yields~\eqref{TrSectComp}.
\end{proof}

\begin{remark}
The first result in Proposition~\ref{prop:RiccComp} is essentially a sharpened version of a standard conjugate point calculation that appears, for example, in the proof of the singularity theorems (see, e.g., \cite[Chapter~4]{HE}). If $c > 0$ (i.e. the Ricci tensor is positive along the null geodesics) then the factor $\sn_c^{\prime}(s)/\sn_c(s)$ diverges to $-\infty$ as $s \to \pi/\sqrt{c}$, which signifies that the geodesic $\gamma_{\l}$ has encountered a conjugate point.

The second result in the Proposition~\ref{prop:RiccComp} implies that if the curvature is bounded above, then either the shape operator is positive definite if $K \le 0$ and positive semi-definite up to affine distance $\pi/\sqrt{K}$ if $K > 0$. A consequence of this is that the geodesics $\gamma_{\l}$ will encounter no conjugate points (if $K \le 0$) or will not encounter them before affine distance $\pi/\sqrt{K}$ (if $K > 0$). The latter result is analogous to a simplified version of the Rauch comparison theorem in Riemannian geometry (see, e.g.,~\cite{CE}). Indeed, our curvature condition that $\mathcal{R}_{\l} \le K$ is equivalent Harris's condition~\cite{Harris} that the null curvature along a null geodesic be bounded above. In four-dimensions, in Newman--Penrose conventions, this curvature bound is equivalent to imposing the condition that $\Phi_{00} + |\Psi_0| \le K$. See the Appendix for more details.
\end{remark}

\section{Comparison results}
\label{sec:comparison}

In this section, we derive our area and volume monotonicity and comparison results.

\begin{theorem}
\label{AreaComparison}
Let $(M, \g)$ be a Lorentzian manifold. Let $p \in M$, and assume that $\Ric(\gammad_{\l}, \gammad_{\l}) \ge c(n-1)$ for along each null generator $\gamma_{\l}$ of $N^+(p)$. Then the area of the cross section of the null cone $\Scal_s$ is such that the map
\bel{mono1}
s \mapsto \frac{|\Scal_s|_{\g}}{\omega_{n-1} \sn_c(s)^{n-1}} \mbox{ is non-increasing}
\ee
and the ratio on the right-hand-side converges to $1$ as $s \to 0$. In particular,
\[
|\Scal_s|_{\g} \le \omega_{n-1} \sn_c(s)^{n-1}.
\]
for $s \ge 0$.
\end{theorem}

If $c = 0$, then $\omega_{n-1} \sn_c(s)^{n-1} \equiv \omega_{n-1} s^{n-1}$ equals the area of the $(n-1)$-sphere of radius $s$. In particular, this is equal to the cross-sectional area of the slice, $S^0_s$, of the null cone in flat Minkowski space. We denote the area of such a slice in Minkowski space by $|S^0_s|_{\eta}$. The final statement in this Theorem therefore allows us to sharpen one of the main results of~\cite{LightCone}:

\begin{theorem}
Let $(M, \g)$ be a Lorentzian metric, the Ricci tensor of which obeys the condition that $\Ric(\gamma', \gamma') \ge 0$ along all future-directed null geodesics from the point $p$. Let $|S^0_s|_{\eta}$ denote the cross-sectional area of the slice $S^0_s$ of the null cone in flat Minkowski space $\R^{n, 1}$. Then the ratio
\bel{LCInequality}
\frac{|\Scal_s|_{\g}}{|S^0_s|_{\eta}}
\ee
is non-increasing as a function of $s$ and converges to $1$ as $s \to 0$. In particular,
\[
|\Scal_s|_{\g} \le |S^0_s|_{\eta}.
\]
\end{theorem}

\begin{proof}[Proof of Theorem~\ref{AreaComparison}]
Equations~\eqref{ddsSs} and~\eqref{RicciComp} imply that
\[
\frac{d}{ds} \log \left( |\Scal_s|_{\g} \right) =
\frac{1}{|\Scal_s|_{\g}} \int_{\Scal_s} \tr\bS \, dV_{\bs_s} \le
(n-1) \frac{\sn_c^{\prime}(s)}{\sn_c(s)}.
\]
Hence
\[
\frac{d}{ds} \log \left( \frac{|\Scal_s|_{\g}}{\sn_c(s)^{n-1}} \right) \le 0,
\]
which yields the monotonicity formula~\eqref{mono1}. To fix the relative normalisation, we note that $|\Scal_s|_{\g}$ and $\omega_{n-1} \sn_c(s)^{n-1}$ both converge to the area of an $(n-1)$-sphere of radius $s$ as $s \to 0$. Therefore, their ratio converges to $1$.
\end{proof}

{}From this result, we may derive an analogue of the Bishop--Gromov comparison result. As is standard, the Lorentzian metric does not induce a semi-Riemannian metric on the null cone $N^+(p)$. We may, however, still define the null volume of the set $\Ncal^+_s(p)$ to be the integral
\[
|\Ncal^+_s(p)|_{\g} := \int_0^s |\Scal_t|_{\g} dt.
\]
For the model quantity, we define
\[
V^+_c(s) := \omega_{n-1} \int_0^s \sn_c(t)^{n-1} \, dt.
\]

Finally, we require the following simple, but surprisingly powerful, observation from~\cite[pp.~42]{CGT}:
\begin{lemma}
\label{CGTLemma}
Let $f, g \colon [0, \infty) \rightarrow (0, \infty)$ with the property that $f/g$ is non-increasing. Then
\[
\frac{\int_0^r f(s)ds}{\int_0^r g(s) ds} \mbox{ is a non-increasing function of $r$}.
\]
\end{lemma}

\

We then have the following result:

\begin{theorem}
\label{thm:volumemonotonicity}
Let $(M, \g)$ be a Lorentzian manifold. Let $p \in M$, and assume that $\Ric(\gammad_{\l}, \gammad_{\l}) \ge c(n-1)$ for along each null generator $\gamma_{\l}$ of $N^+(p)$. Then the null volume of the set $\Ncal^+_s(p)$ is such that the map
\bel{mono2}
s \mapsto \frac{|\Ncal^+_s(p)|_{\g}}{V^+_c(s)} \mbox{ is non-increasing}
\ee
and the ratio on the right-hand-side converges to $1$ as $s \to 0$. In particular,
\[
|\Ncal^+_s(p)|_{\g} \le V^+_c(s)
\]
for $s \ge 0$.
\end{theorem}
\begin{proof}
Taking $f(s) = |\Scal_s|_{\g}$ and $g(s) = \omega_{n-1} \sn_c(s)^{n-1}$, then Theorem~\ref{AreaComparison} implies that the ratio $f/g$ is non-increasing. Applying Lemma~\ref{CGTLemma} then gives the monotonicity result~\eqref{mono2}. Again, the limiting value of the ratio as $s \to 0$ is clearly $1$.
\end{proof}

\

On the other hand, if we assume an upper bound on the null curvature, then we derive a dual version of the area monotonicity formula:

\begin{theorem}
Let $(M, \g)$ be a Lorentzian manifold. Let $p \in M$, and assume that $\mathcal{R}_{\l} \le K$ for along each null generator $\gamma_{\l}$ of $N^+(p)$. Then the area of the cross section of the null cone $\Scal_s$ is such that the map
\[
s \mapsto \frac{|\Scal_s|_{\g}}{\omega_{n-1} \sn_K(t)^{n-1}} \mbox{ is non-decreasing}
\]
and the ratio on the right-hand-side converges to $1$ as $s \to 0$. In particular,
\[
|\Scal_s|_{\g} \ge \omega_{n-1} \sn_K(t)^{n-1}
\]
for $s \ge 0$.
\end{theorem}
\begin{proof}
The proof exactly parallels that of Theorem~\ref{AreaComparison}, but we use the inequality~\eqref{TrSectComp}, rather than~\eqref{RicciComp}.
\end{proof}

\begin{remark}
This area monotonicity theorem is, essentially, the opposite of the result of~\cite{LightCone} and Theorem~~\ref{AreaComparison}, giving a \emph{lower\/} bound on the area of the section of the null cone. Note, however, that the curvature condition required is an \emph{upper bound\/} on the curvature operator $\mathcal{R}_{\l}$ along the null geodesics, which is considerably stronger than, for example, an upper bound on the Ricci tensor. The fact that we require a stronger type of curvature bound for this type of theorem is familiar from similar considerations in Riemannian geometry.
\end{remark}

Finally, we have the following analogue of the G\"{u}nther volume comparison theorem:
\begin{theorem}
Let $(M, \g)$ be a Lorentzian manifold. Let $p \in M$, and assume that $\mathcal{R}_{\l} \le K$ for along each null generator $\gamma_{\l}$ of $N^+(p)$. Then the null volume of the set $\Ncal^+_s(p)$ satisfies
\[
|\Ncal^+_s(p)|_{\g} \ge V^+_K(s)
\]
for $s \ge 0$.
\end{theorem}
\begin{proof}
\[
\frac{d}{ds} |\Ncal^+_s(p)|_{\g} = |S^+(s)|_{\g} \ge \omega_{n-1} \sn_K(t)^{n-1} = \frac{d}{ds} V^+_K(s).
\]
Moreover, $|\Ncal^+_s(p)|_{\g}$ and $V^+_K(s)$ both converge to the null volume of the corresponding subset of the null cone in Minkowski space as $s \to 0$, so their ratio converges to $1$ as $s \to 0$.
\end{proof}

\section{Application to Ricci-flat four-manifolds}
\label{sec:4d}

We briefly outline a simple consequence of our results for the special case of four-dimensional Lorentzian manifolds that satisfy the vacuum Einstein condition $\Ric = 0$. We first define the \lq\lq area radius\rq\rq\ of the sphere $\Scal_s$ by the equality
\[
r(s) := \sqrt{\frac{|\Scal_s|_{\g}}{4 \pi}}.
\]
If one wishes to measure the deviation of properties of the null cone from that in flat Minkowski space, then a standard quantity that one must estimate\footnote{See, e.g.,~\cite{KR:Inventiones} for an analytical investigation of this and related objects in a low-regularity setting.} is the difference
\[
\tr\bS - \frac{2}{r(s)}.
\]

Our results give the following, simple estimate:

\begin{proposition}
\label{prop:arearadius}
Let $(M, \g)$ be a Ricci-flat Lorentzian four-manifold. Let $p \in M$ and $K \ge 0$ a constant such that, along the null geodesics $\gamma_{\l}$ emanating from $p$, the curvature operator $\mathcal{R}_{\l}$ satisfies the condition that $\mathcal{R}_{\l} \le K$. Then, for all $s > 0$, we have
\bel{arearadRicci}
\tr\bS - \frac{2}{r(s)} \le 0
\ee
and, for $0 \le s < \pi/\sqrt{K}$,
\bel{arearadSec}
\tr\bS - \frac{2}{r(s)} \ge - 2 \sqrt{K} \tan \left( \frac{\sqrt{K}}{2} s\right).
\ee
\end{proposition}

\begin{remark}
\label{rem:arearad}
Equation~\eqref{arearadRicci} shows that the mean curvature of the null slices for a cone in a Ricci-flat is bounded above by the flat-space expression in terms of the area radius. It follows from~\eqref{arearadSec} that, if $\g$ is Ricci-flat and the curvature operator is bounded above, then, given any $\eps > 0$, then there exists $s_0 > 0$ such that
\[
- \eps \le \tr\bS - \frac{2}{r_V(s)} \le 0 \mbox{ for $s \le s_0$}.
\]
As such, for such manifolds, we may put explicit bounds on the deviation of $\tr\bS$ from the flat-space expression in terms of the area radius, for small $s$.
\end{remark}

\begin{proof}[Proof of Proposition~\ref{prop:arearadius}]
Since $\g$ is Ricci-flat, we may take $c = 0$ in our Ricci curvature bound. Proposition~\ref{prop:RiccComp} then yields the inequalities
\[
2 \frac{\sn^{\prime}_K(s)}{\sn_K(s)} \le \tr\bS(s) \le \frac{2}{s}.
\]
Our area comparison results, in addition, imply that
\[
4 \pi \, \sn_K(s)^2 \le |\Scal_s|_{\g} \le 4 \pi s^2.
\]
Therefore the area radius satisfies
\[
|\sn_K(s)| \le r(s) \le s.
\]

We therefore have
\[
\tr\bS \le \frac{2}{s} \le \frac{2}{r(s)},
\]
giving the second of our required inequalities. In addition,
\begin{align*}
\tr\bS - \frac{2}{r(s)}
&\ge 2 \frac{\sn^{\prime}_K(s)}{\sn_K(s)} - \frac{2}{\sn_K(s)}
= 2 \sqrt{K} \frac{\cos(\sqrt{K}s) - 1}{\sin(\sqrt{K}s)}
\\
&= - 2 \sqrt{K} \tan \left( \frac{\sqrt{K}}{2} s\right),
\end{align*}
as required.
\end{proof}

\vs
When considering lower bounds on Ricci curvature, it is perhaps volume monotonicity that plays a more important role than area comparison theorems. Therefore, we define the \lq\lq volume radius\rq\rq\ of the set $\Ncal^+_s(p)$ by the relation
\[
r_V(s) := \left( \frac{3 |\Ncal^+_s(p)|_{\g}}{4 \pi} \right)^{1/3}.
\]
Our volume comparison theorems then state that
\[
|\sn_K(s)| \le r_V(s) \le s.
\]
Therefore, with an identical proof to the previous Proposition, we have the following result:

\begin{proposition}
\label{prop:volumeradius}
Let $(M, \g)$ be a Ricci-flat Lorentzian four-manifold. Let $p \in M$ and $K \ge 0$ a constant such that, along the null geodesics $\gamma_{\l}$ emanating from $p$, the curvature operator $\mathcal{R}_{\l}$ satisfies the condition that $\mathcal{R}_{\l} \le K$. Then, for all $s > 0$, we have
\[
\tr\bS - \frac{2}{r_V(s)} \le 0
\]
and, for $0 \le s < \pi/\sqrt{K}$,
\[
\tr\bS - \frac{2}{r_V(s)} \ge - 2 \sqrt{K} \tan \left( \frac{\sqrt{K}}{2} s\right).
\]
\end{proposition}

\begin{remark}
It follows from this Proposition that the observations made in Remark~\ref{rem:arearad} concerning the area radius also hold true for the volume radius.
\end{remark}

\begin{remark}
Clearly, our result only actually requires that the Ricci curvature of $\g$ be non-negative along the null geodesics $\gamma_{\l}$. Our results also generalise to arbitrary dimension in the obvious fashion.
\end{remark}

\section{Model spaces}
\label{sec:model}
In Riemannian geometry, comparison theorems generally compare a geometrical quantity (e.g. volumes and areas of sets) on a manifold that satisfies a curvature bound with corresponding quantities in a model space of, for example, constant curvature. Before studying the model geometries that one should use for comparison in our theorems, we first note the following simple facts:
\begin{enumerate}
\item Let $(M_c, \g_c)$ denote the simply-connected, $n$-dimensional Riemannian manifold of constant curvature $c$. Given $p \in M_c$, the area of the distance sphere $S(p, s)$ is equal to $\omega_{n-1} \sn_c(t)^{n-1}$. We denote this quantity by $S_c(s)$.
\item In the same space, the volume of the distance ball $B(p, s)$ is equal to the quantity $V^+_c(s)$. We denote this quantity by $V_c(s)$.
\end{enumerate}

Our comparison theorems may therefore be restated as giving comparison results between areas of spherical slices of a null cone in $(n+1)$-dimensional Lorentzian manifolds and spheres in $n$-dimensional constant curvature spaces, and corresponding volumes in $(n+1)$-dimensional Lorentzian manifolds, and the corresponding quantities in $n$-dimensional constant curvature Riemannian manifolds:

\begin{theorem}
Let $(M, \g)$ be a Lorentzian metric, the Ricci tensor of which obeys the condition that $\Ric(\gamma', \gamma') \ge 0$ along all future-directed null geodesics from the point $p$. Then the ratios
\[
\frac{|\Scal_s|_{\g}}{S_c(s)}, \qquad \frac{|\Ncal^+_s(p)|_{\g}}{V_c(s)}
\]
are non-increasing as functions of $s$ and converge to $1$ as $s \to 0$.

Similarly, let $(M, \g)$ be a Lorentzian manifold such that $\mathcal{R}_{\l} \le K$ for along each null generator $\gamma_{\l}$ of $N^+(p)$. Then the ratio
\[
s \mapsto \frac{|\Scal_s|_{\g}}{S_K(s)}
\]
is non-decreasing and converges to $1$ as $s \to 0$. In addition, the null volume of the set $\Ncal^+_s(p)$ satisfies
\[
|\Ncal^+_s(p)|_{\g} \ge V_K(s)
\]
for $s \ge 0$.
\end{theorem}

\subsection{Lorentzian model spaces}
Although stating our results in terms of comparison with Riemannian constant curvature spaces is of interest, it would be more fitting to state our results as comparing areas of slices of null cones with, for example, corresponding slices of cones in a model Lorentzian manifold. Therefore, we now briefly consider model Lorentzian manifolds on which our estimates are sharp. Based upon our different curvature bounds, there are two types of model spaces that we should naturally consider. Firstly, we consider Lorentzian manifolds where we have $\Ric(\gammad_{\l}, \gammad_{\l}) = c (n-1)$ along the null geodesics from a given point $p$ in the manifold, and where the differential inequality satisfied by the mean curvature~\eqref{Hinequality} becomes an equality. Secondly, we consider Lorentzian manifolds where the curvature operator $\mathcal{R}_{\l}$ equals $K \,\Id$ along such null geodesics. Note that we cannot expect these conditions to uniquely determine a model geometry since, for example, the Ricci curvature condition with $c = 0$ is satisfied by all of the constant curvature spaces.

\vs
Our first result is that the latter class of model spaces includes the former:
\begin{lemma}
Let $(M, \g)$ satisfy the curvature equality
\bel{RicEq}
\Ric(\gammad_{\l}, \gammad_{\l}) \ge c (n-1),
\ee
and the equality
\bel{Hequality}
H'(s) + H(s)^2 = - c
\ee
along all null geodesics from $p \in M$. Then, along the same geodesics, the curvature operator satisfies
\bel{Rl}
\mathcal{R}_{\l} = c \, \Id,
\ee
\end{lemma}
\begin{proof}
Taking the trace of the Riccati equation~\eqref{SRiccati}, we deduce that
\[
H' = - H^2 - \frac{1}{n-1} \left[ \tr (\sigma^2) + \Ric(\gammad_{\l}, \gammad_{\l}) \right],
\]
where $\sigma := \bS - H \,\Id$ denotes the trace-free part of the shape operator. Since we have, by assumption, that $H' = - H^2 - c$ and $\Ric(\gammad_{\l}, \gammad_{\l}) \ge c(n-1)$, it follows that $\Ric(\gammad_{\l}, \gammad_{\l}) = c(n-1)$ and $\tr(\sigma^2) = 0$. This implies that $\sigma = 0$, and therefore that
\[
\bS = H \,\Id.
\]
Moreover, the differential equation~\eqref{Hequality} along with the asymptotic conditions on $H(s)$ as $s \to 0$, imply that
\[
H(s) = \frac{\sn^{\prime}_c(s)}{\sn_c(s)}.
\]
Therefore,
\[
\nabla_{\l} \bS + \bS^2 = - c \, \Id,
\]
as required.
\end{proof}

The fact that we wish the trace of the second-fundamental form to vanish suggests that we consider (locally) conformally flat manifolds. For conformally flat metrics, all of the curvature information is contained in the Ricci tensor and one may easily check that, if the metric $\g$ is conformally flat and the Ricci tensor satisfies~\eqref{RicEq}, then the curvature operator takes the form~\eqref{Rl}.

As mentioned earlier, our curvature condition will not lead to a unique model geometry with which we should compare. As such, our comparison results will not, in general, directly lead to a rigidity condition if the estimates are sharp.\footnote{Rigidity results were derived in~\cite{LightCone}, when additional conditions were imposed.} Since we have no unique model geometry, we simply present some Lorentzian metrics that have the required properties.

\begin{example}
Let $\g_{S^{n-1}}$ denote the standard metric on the unit $(n-1)$-sphere. We then define $(M_c, \g_c)$ as follows:
\[
\g_c := \frac{1}{1 + c t^2} \left[ dt^2 + dr^2 + r^2 \g_{S^{n-1}} \right],
\]
where the coordinates $(t, r)$ lie in the range:
\begin{itemize}
\item $t, r < -\tfrac{1}{\sqrt{|c|}}$ if $c < 0$;
\item $t, r \in \R$ if $c \ge 0$.
\end{itemize}
Taking the reference point, $p_c$, to be the origin $t = r = 0$, and the reference vector $\T_c = \partial_t \in T_{p_c} M$, then it is straightforward to check that
\[
\Scal_s = \left\{ t = r = \frac{\sn_c(s)}{\sn^{\prime}_c(s)} \right\},
\]
with induced metric
\[
\bs_s = \sn_c(s)^2 \g_{S^{n-1}}.
\]
Therefore, $|\Scal_s|_{\g_c} = \omega_{n-1} \, \sn_c(s)^{n-1}$ and $|\Ncal^+_s(p)|_{\g_c} = \omega_{n-1} \int_0^s \sn_c^{n-1}(t) \, dt$. Moreover, the mean curvature of the sphere $\Scal_s$ is $H = \frac{\sn^{\prime}_c(s)}{\sn_c(s)}$, as required. Finally, letting $\gamma_{\l}$ be the affinely parametrised null geodesic with respect to the metric $\g_c$, then a standard curvature calculation shows that
\[
\Ric(\gammad_{\l}, \gammad_{\l}) = c (n-1).
\]
\end{example}

\

The pointed Lorentzian manifolds with reference vector $(M_c, \g_c, p_c, \T_c)$ thus defined have the correct properties to be viewed as model geometries for our comparison theorems.\footnote{As mentioned above, however, they are not unique in this respect.} As such, we may reformulate our comparison and monotonicity results in the following fashion.

Let $(M, \g)$ be a Lorentzian manifold, $p \in M$ and $\T \in T_p M$ a reference future-directed, unit, time-like vector. Let $c$ be a real constant such that $\Ric(\gammad_{\l}, \gammad_{\l}) \ge c(n-1)$ along the future-directed null geodesics from $p$. Given the reference model $(M_c, \g_c, p_c, \T_c)$ defined as above, let $\varphi \colon T_p M \to T_{p_c} M_c$ be a linear isometry with the property that $\varphi_* \T = \T_c$. For sufficiently small $s > 0$, given the sets $\Scal_s, \Ncal^+_s(p) \subset M$, the \lq\lq transplantation\rq\rq\ map\footnote{We follow the terminology of~\cite{ES}.} $\widetilde{\varphi} := \exp_{p_c} \circ \varphi \circ \exp_p^{-1}$ allows us to define corresponding subsets $\widetilde{\varphi}(\Scal_s), \widetilde{\varphi}(\Ncal^+_s(p))$ in the manifold $M_c$. (Since the map $\varphi$ is an isometry, these are the same sets as we would get by applying the constructions in Section~\ref{sec:basic} to the Lorentzian manifold $(M_c, \g_c)$ based at the point $p_c$ with reference vector $\T_c$.) We denote the area and volume of these subsets of $M_c$ by $|\widetilde{\varphi}(\Scal_s)|_c$ and $|\widetilde{\varphi}(\Ncal^+_s(p))|_c$, respectively. In precisely the same fashion, we may construct a similar map from a Lorentzian manifold $(M, \g)$ satisfying $\mathcal{R}_{\l} \le K$ to the model space $(M_K, \g_K)$. Our results may then be recast as follows:

\begin{theorem}
Let $(M, \g)$ be a Lorentzian metric, the Ricci tensor of which obeys the condition that $\Ric(\gamma', \gamma') \ge 0$ along all future-directed null geodesics from the point $p$. Let $(M_c, \g_c, p_c, \T_c)$ be the model space as above, and $\widetilde{\varphi}$ the corresponding transplantation map. Then the ratios
\[
\frac{|\Scal_s|_{\g}}{|\widetilde{\varphi}(\Scal_s)|_c}, \qquad
\frac{|\Ncal^+_s(p)|_{\g}}{|\widetilde{\varphi}(\Ncal^+_s(p))|_c}
\]
are non-increasing as functions of $s$ and converge to $1$ as $s \to 0$.

Similarly, let $(M, \g)$ be a Lorentzian manifold such that $\mathcal{R}_{\l} \le K$ for along each null generator $\gamma_{\l}$ of $N^+(p)$. Let $(M_K, \g_K, p_K, \T_K)$ be the corresponding model space. Then the ratio
\[
s \mapsto \frac{|\Scal_s|_{\g}}{|\widetilde{\varphi}(\Scal_s)|_c}
\]
is non-decreasing and converges to $1$ as $s \to 0$. In addition, the null volume of the set $\Ncal^+_s(p)$ satisfies
\[
|\Ncal^+_s(p)|_{\g} \ge |\widetilde{\varphi}(\Ncal^+_s(p))|_c
\]
for $s \ge 0$.
\end{theorem}

\section{Final remarks}
\label{sec:remarks}

We have implicitly assumed in our analysis that we are considering values of $s$ less than the null injectivity radius at $p$, so that the exponential map defines a global diffeomorphism between an open neighbourhood of a subset of the null cone $T_p M$ and a corresponding open neighbourhood of a subset of the null cone of $p$ in $M$. Recall that a null geodesic $\gamma_{\l}$ from a point $p$ in a geodesically complete Lorentzian manifold will be maximising until the cut point $\gamma_{\l}(s_0)$ where either $\gamma_{\l}(s_0)$ is conjugate to $p$ along $\gamma_{\l}$ or there exists a distinct null geodesic from $p$ that also passes through $\gamma_{\l}(s_0)$. For $s > s_0$, there exists a time-like geodesic from $p$ to the point $\gamma_{\l}(s)$, so $\gamma_{\l}(s)$ no longer lies on the boundary of the causal future of $p$. In line with Gromov's approach to volume monotonicity theorems~\cite{Gromov:Book}, our volume monotonicity result Theorem~\ref{thm:volumemonotonicity} may be extended past the null injectivity radius by cutting off the volume integral once our null geodesics intersect the null cut locus of $p$. Such a truncation of the volume integral will, generally, decrease the volume integral in the numerator of the ratio $|\Ncal^+_s(p)|_{\g} / V^+_c(s)$, and will therefore strengthen the monotonic behaviour.

Our results may be generalised in an obvious fashion to apply open subsets of the space of null directions at $p$, in particular null neighbourhoods of a given null geodesic. If we wish to lower the regularity of our metric $\g$ then, in the usual spirit of synthetic geometry, one could adopt our volume monotonicity theorem as a \emph{definition\/} of lower Ricci curvature bounds in null directions. It would be particularly interesting to know whether one could, for example, prove a version of the Penrose singularity theorem or positivity of the Bondi mass with this definition of a lower bound on the Ricci curvature. For a definition of lower and upper curvature bounds in the sense of bounds on our operator $\mathcal{R}_{\l}$ then, by analogy with the theory of Alexandrov spaces, it would probably be more appropriate to base such a definition on a Lorentzian version of the Toponogov comparison theorem, such as that discussed in~\cite{AB}.

\appendix
\section{Newman--Penrose formalism}
\label{sec:NP}

We now briefly show how, in four dimensions, we may carry out all of our calculations in Newman--Penrose formalism. Note that, unlike the body of this paper, this section is not self-contained. Background material on Newman--Penrose formalism may be found in, for example,~\cite[Chapter~4]{PR}.

\vs
In Newman--Penrose formalism, the fact that $\nabla_{\l} \l = 0$ implies that $\kappa = 0$ and $\eps + \overline{\eps} = 0$. Imposing that $\left[ \m, \mbar \right]$ has no $\n$ component is equivalent to imposing that $\rho$ be real, while imposing that it have no $\l$ component is equivalent to reality of $\rho'$. Changing $\m$ and $\mbar$ by a phase, we may impose that $\nabla_{\l} \m \propto \l$ and $\nabla_{\l} \mbar \propto \l$, which implies that $\eps - \overline{\eps} = 0$. This completely fixes the basis vectors $\l, \m$ and $\mbar$. The vector field $\n$ on $\Ncal^+_s(p) \setminus \{ p \}$ is then uniquely determined by the null orthogonality conditions. We may, therefore, assume that spin-coefficients satisfy
\[
\kappa = 0, \quad \eps = 0, \quad \overline{\rho} = \rho, \quad \overline{\rho'} = \rho'.
\]
The Newman--Penrose equations that we require are
\begin{subequations}
\begin{align}
\frac{d}{ds} {\rho} &= {\rho}^2 + {\sigma} {\overline{\sigma}} + \Phi_{00}
\label{NPrho}
\\
\frac{d}{ds} {\sigma} &= 2 \rho \sigma + \Psi_{0}
\label{NPsigma}
\end{align}
\label{NP}
\end{subequations}

\subsection{Minkowski space}
When calculating geometrical quantities related to the spheres $\Scal_s$ in an arbitrary Lorentzian manifold, we will need to fix various constants that appear by comparing the asymptotic behaviour as $s \to 0$ to the values of the corresponding quantities in flat Minkowski space. We therefore summarise, here, the values of all relevant quantities in Minkowski space. Here, we would take
\begin{align*}
\l &= \frac{1}{\sqrt{2}} \left( \partial_t + \partial_r \right),
&\n &= \frac{1}{\sqrt{2}} \left( \partial_t - \partial_r \right),
\\
\m &= \frac{1}{\sqrt{2} r} \left( \partial_{\theta} - \frac{i}{\sin\theta} \partial_{\phi} \right),
&\mbar &= \frac{1}{\sqrt{2} r} \left( \partial_{\theta} + \frac{i}{\sin\theta} \partial_{\phi} \right).
\end{align*}
In $(t, r, \theta, \phi)$ coordinates, the geodesic $\gamma_{\l}$ then takes the form
\[
\gamma_{\l}(s) = \left( \frac{s}{\sqrt{2}}, \frac{s}{\sqrt{2}}, \theta_0, \phi_0 \right).
\]
The set $\Scal_s$ is then the set $\{ t = r = \frac{s}{\sqrt{2}} \}$, with induced metric
\[
\bs_s = \frac{1}{2} s^2 \left( d\theta^2 + \sin^2 \theta d\phi^2 \right).
\]
Note that we therefore have
\[
|\Scal_s|_{\g} = 2 \pi s^2.
\]
The spin-coefficients that are of concern to us take the form
\[
\rho(s) = - \frac{1}{s}, \qquad \sigma(s) = 0.
\]

\subsection{Manifolds with curvature bounds}
In the case where the Ricci coefficient $\Phi_{00}$ is bounded below, then we may treat equation~\eqref{NPrho} by scalar Riccati techniques.
\begin{proposition}
If $\Phi_{00} \ge c$, then
\bel{rholowerbound}
\rho(s) \ge - \frac{\sn_c^{\prime}(s)}{\sn_c(s)}.
\ee
\end{proposition}
\begin{proof}
If $\Phi_{00} \ge c$, then, denoting $\tfrac{d}{ds}$ by ${}'$, we have, from~\eqref{NPrho},
\[
\rho' \ge \rho^2 + c.
\]
Letting $\rho(s) = - a'(s) / a(s)$, where $a(0) = 0$, $a'(0) = 1$, then we have
\[
a^{\prime\prime}(s) + c a(s) \le 0.
\]
Therefore,
\[
\frac{a'(s)}{a(s)} \le \frac{\sn_c^{\prime}(s)}{\sn_c(s)},
\]
and, hence,
\[
\rho(s) \ge - \frac{\sn_c^{\prime}(s)}{\sn_c(s)}.
\]
\end{proof}

Alternatively, we may treat equations~\eqref{NP} together by matrix Riccati techniques if we have an upper bound on the curvature. As in~\cite[Chapter~7]{PR}, we define the $2 \times 2$ complex matrices
\[
P :=
\begin{pmatrix} \rho &\sigma \\ \overline{\sigma} &\rho
\end{pmatrix}, \qquad
Q :=
\begin{pmatrix} \Phi_{00} &\Psi_0 \\ \overline{\Psi}_0 &\Phi_{00}
\end{pmatrix}.
\]
We then have the following:

\begin{proposition}
Let $\Lambda$ be a real constant such that
\bel{Qbound}
\Phi_{00} + |\Psi_0| \le \Lambda.
\ee
Then
\bel{PUpperBound}
P(s) \le - \frac{\sn_{\Lambda}^{\prime}(s)}{\sn_{\Lambda}(s)} \,\Id,
\ee
in the sense that the eigenvalues of the operator $P(s)$ are bounded above by $- \frac{\sn_{\Lambda}^{\prime}(s)}{\sn_{\Lambda}(s)}$, i.e.
\[
\rho(s) \pm |\sigma(s)| \le -
\frac{\sn_{\Lambda}^{\prime}(s)}{\sn_{\Lambda}(s)}.
\]
\end{proposition}

\begin{proof}
The Newman--Penrose equations~\eqref{NPrho} and~\eqref{NPsigma} may be written in the form
\bel{MatrixRiccP}
\frac{d}{ds}P = P^2 + Q.
\ee
Since the matrices $P$ and $Q$ are Hermitian with respect to the standard inner product on $\mathbb{C}^2$, they both have real eigenvalues. The proof then follows the same strategy as in the second part of Proposition~\ref{prop:RiccComp}.
\end{proof}

\subsection{Areas}
For $s > 0$, we denote the area of the sphere $\Scal_s$ with respect to the volume form defined by $\bs_s$ by
\[
|\Scal_s|_{\g} := \int_{\Scal_s} \, dV_{\bs_s}.
\]
We then have
\[
\frac{d}{ds} |\Scal_s|_{\g} = - 2 \int_{\Scal_s} \rho \, dV_{\bs_s}.
\]

\begin{corollary}
Let $c$ be a constant such that $\Phi_{00} \ge c$. Then the ratio
\[
|\Scal_s|_{\g}/(2 \pi\,\sn_c(s)^2)
\]
is non-increasing as a function of $s$, and converges to $1$ as $s \to 0$. In particular,
\bel{VolumeInequality1}
|\Scal_s|_{\g} \le 2 \pi \,\sn_c(s)^2.
\ee
Similarly, let $\Lambda$ be a constant such that~\eqref{Qbound} holds. Then $|\Scal_s|_{\g}/(2 \pi\,\sn_{\Lambda}(s)^2)$ is non-decreasing as a function of $s$, and converges to $1$ as $s \to 0$. Therefore, we have
\bel{VolumeInequality2}
|\Scal_s|_{\g} \ge 2 \pi \,\sn_{\Lambda}(s)^2.
\ee
\end{corollary}
\begin{proof}
We have
\[
\frac{d}{ds} |\Scal_s|_{\g} = - 2 \int_{\Scal_s} \rho \, dV_{\bs_s} \le 2
\frac{\sn_c^{\prime}(s)}{\sn_c(s)} |\Scal_s|_{\g}.
\]
Therefore,
\[
\frac{d}{ds} \log \left( \frac{|\Scal_s|_{\g}}{\sn_c(s)^2} \right) \le 0,
\]
and, hence, the ratio $|\Scal_s|_{\g}/\sn_c(s)^2$ is non-increasing. As $s \to 0$, $|\Scal_s|_{\g}$ approaches the flat-space value $2 \pi s^2$ and $\frac{1}{s} \sn_c(s) \to 1$, so
\[
\lim_{s \to 0} \frac{|\Scal_s|_{\g}}{\sn_c(s)^2} = 2 \pi.
\]
Combining the monotonicity result with this limiting result gives~\eqref{VolumeInequality1}.

A similar argument, using the fact that $\rho \le - \sn_{\Lambda}^{\prime}(s) / \sn_{\Lambda}(s)$, gives the second result.
\end{proof}

As in the main part of the paper, the volume results then follow directly from the area monotonicity properties.

\end{document}